\documentclass[11pt,a4paper]{article}
\usepackage{latexsym,amssymb,amsmath,amsthm,braket,color}
\usepackage[ruled,vlined]{algorithm2e}
\usepackage{ascmac,fancybox, bm}
\usepackage{calc}
\usepackage{mathrsfs}
\usepackage{setspace}
\usepackage{graphicx}

\newtheorem{theorem}{Theorem}[section]
\newtheorem{lemma}[theorem]{Lemma}  
\newtheorem{proposition}[theorem]{Proposition} 
\newtheorem{corollary}[theorem]{Corollary}     

\newtheorem{claim}[theorem]{Claim}

\theoremstyle{definition}
\newtheorem{definition}[theorem]{Definition}
\newtheorem{example}[theorem]{Example}
\newtheorem{remark}[theorem]{Remark}

\newcommand{\Z}{\mathbf{Z}}

\newcommand{\Cm}{C_{\!\M}}

\newcommand{\F}{\mathcal{F}}

\newcommand{\J}{\mathcal{J}}
\newcommand{\I}{\mathcal{I}}

\newcommand{\laminar}{\mathcal{L}}
\newcommand{\zero}{{\bf 0}}

\newcommand{\M}{\mathscr{M}}
\newcommand{\C}{\mathcal{C}}

\begin{document}
\title{Envy-free Matchings with Lower Quotas}
\author{
Yu Yokoi
\thanks{National Institute of Informatics, Tokyo 101-8430, Japan. 
E-mail: {\tt yokoi@nii.ac.jp}.}
}
\date{\empty}
\maketitle
\setstretch{1.1}
\vspace{-3mm}
\begin{abstract}
	While every instance of the Hospitals/Residents problem admits a stable matching,
	the problem with lower quotas (HR-LQ) has instances with no stable matching.
	For such an instance, we expect the existence of an envy-free matching, 
	which is a relaxation of a stable matching preserving a kind of fairness property.
	
	In this paper, we investigate the existence of an envy-free matching in several settings,
	in which hospitals have lower quotas and not all doctor-hospital pairs are acceptable. 
	We first provide an algorithm that decides whether a given HR-LQ instance has an envy-free matching or not. 
	Then, we consider envy-freeness in the Classified Stable Matching model due to Huang (2010),
	i.e., each hospital has lower and upper quotas on subsets of doctors.
	We show that, for this model, 
	deciding the existence of an envy-free matching is \mbox{NP-hard} in general, 
	but solvable in polynomial time if quotas are paramodular. 
\end{abstract}

\section{Introduction}
Since the seminal work of Gale and Shapley \cite{GS62}, 
the {\em Hospitals/Residents problem} (HR, for short), or the {\em College Admission problem}, has been studied extensively 
\cite{GIbook, Manlovebook, RS92}.                                                    
They proposed an algorithm that finds a stable matching in linear time for every instance. 
In this problem, each hospital has an upper quota for the number of doctors assigned to it.
In some applications, each hospital also has a lower quota for the number of doctors it receives.
That is, we want to consider the Hospitals/Residents problem with lower quotas (HR-LQ, for short). 
Unfortunately, for HR-LQ, we cannot ensure the existence of a stable matching.
However, it is easy to decide whether there is a stable matching or not for a given HR-LQ instance, 
because the number of doctors assigned to each hospital is identical for any stable matching 
(according to the well-known Rural Hospitals Theorem \cite{GS85, Roth84b, Roth84, Roth86}).

When a given HR-LQ instance has no stable matching,
one natural approach is to weaken stability concept while preserving some kind of fairness.
{\em Envy-freeness} \cite{WR16} (also called {\em fairness} in the school choice literature \cite{Fragiadakis16, Goto16}) 
of matchings is a relaxation of stability obtained by giving up efficiency.
Similarly to stability, envy-freeness forbids the existence of a doctor who has justified envy toward some other doctor,
but it tolerates the existence of a doctor who claims a hospital's vacant seat.
The importance of envy-freeness and its variants has recently been 
recognized in the context of constrained matching 
\cite{Fragiadakis16, Goto16, KK14, KK16, Ehlers14}, 
and structural properties of envy-free matchings were investigated in \cite{WR16}.

Envy-free matchings naturally arise when we find a matching in the following ad hoc manner. 
For an HR-LQ instance, suppose that we find a stable matching while disregarding the lower quotas, 
and that the obtained matching does not meet the lower quotas. 
Let us reduce the upper quotas of hospitals that receive many doctors, 
and again find a stable matching while disregarding the lower quotas, and repeat.
If we find a stable matching that meets the lower quotas after repeating such adjustments, 
then the obtained matching is an envy-free matching of the original instance (see Proposition~\ref{prop:HRtruncate}).

Because an envy-free matching is a relaxation of a stable matching, it is more likely to exist. 
Indeed, if all doctor-hospital pairs are acceptable and the sum of lower quotas of all hospitals does not exceed the number of doctors, 
then we can ensure the existence of an envy-free matching. 
(This follows from the results of Fragiadakis et al. \cite{Fragiadakis16}).
However, if not all pairs are acceptable, then even an envy-free matching may fail to exist. 
Moreover, deciding the existence of an envy-free matching 
is not so simple because envy-free matchings have different sizes
unlike stable matchings.
\vspace{-3mm}
\paragraph*{Our Contribution}
In this paper, we study envy-free matchings for the HR-LQ model and its generalizations.
In our models, not all doctor-hospital pairs are acceptable (i.e., preference lists are incomplete).

We first investigate envy-free matchings in the setting of HR-LQ.
We provide the following characterization of the existence of an envy-free matching.
Let $I$ be a given HR-LQ instance and let $I'$ be an HR instance obtained from $I$ by removing lower quotas and replacing upper quotas with the original lower quotas. 
We prove that $I$ has an envy-free matching if and only if every hospital is full in a stable matching of $I'$ (Theorem~\ref{thm:HR1}).
Combined with the rural hospitals theorem, this characterization yields an efficient algorithm
to decide the existence of an envy-free matching for an HR-LQ instance.
That is, we can decide it by finding a stable matching for the HR instance whose upper quotas are 
the original lower quotas,
and checking whether all hospitals are full or not.

Next, we move to a generalized model, in which each hospital imposes an upper and a lower quota on each subset of doctors.
That is, we consider an envy-free matching version of Huang's {\em Classified Stable Matching} \cite{CCH10} (CSM, for short).
(See ``Related Works'' below for results on stable matchings of CSM and its generalizations.)
In Huang's original model, each hospital has a family of sets of doctors, called {\em classes}, and each class has an upper and a lower quota.
We formulate this setting by letting each hospital have a pair of set functions defined on the set of acceptable doctors. 
These two functions respectively represent upper quotas and lower quotas.
For this model, we show that it is NP-hard to decide the existence of an envy-free matching, 
even if the number of non-trivial quotas is linear (Theorem~\ref{thm:HR1}).
The proof is by a reduction from the NP-complete problem (3,B2)-SAT \cite{BKS03}.

Then, we provide a tractable special case of CSM.
We show that if the pair of lower and upper quota functions of each hospital is {\em paramodular} \cite{Frankbook}
(see Section~\ref{sec:g-mat} for the definition), 
then we can decide the existence of an envy-free matching in polynomial time.
This means that the problem is tractable if the family of acceptable doctor sets forms a generalized matroid for each hospital.
A {\em generalized matroid} \cite{Tardos85} (also called an {\em M$^{\natural}$-convex family} \cite{Murota16survey}) 
is a family of subsets satisfying a certain axiom called the exchange axiom.
It is known that a paramodular function pair defines a generalized matroid and vice versa.
Because constraints defined on a laminar (or hierarchical) family yield a generalized matroid,
our tractable special case includes a case in which each hospital defines quotas on a laminar family of doctors.
\vspace{-3mm}
\paragraph*{Related Works}
Recently, the study of matching models with lower quotas has developed substantially 
\cite{ACGMM16, FK16, Goto16, HIM11, HIM16, CCH10, Manlovebook, Mnich16}.
The Hospitals/Residents problem with lower quotas (HR-LQ) was first studied by Hamada et al. \cite{HIM11, HIM16},
who considered the minimization of the number of blocking pairs subject to upper and lower quotas. 
They showed the NP-hardness of the problem, gave an inapproximability result, and provided an exponential-time exact algorithm. 
Motivated by the matching scheme used in  Hungary's higher education sector,
Bir\'{o} et al. \cite{BFIM10} considered a version of HR-LQ in which hospitals (i.e., colleges) 
are allowed to be closed, i.e., each hospital is assigned enough doctors or no doctor. 
They showed the NP-completeness to decide the existence of a stable matching.

The Classified Stable Matching problem (CSM), proposed by Huang \cite{CCH10}, is a generalization of HR-LQ without hospital closures.
In this model, each hospital (or institute, in Huang's terminology) has a classification of doctors 
(i.e., applicants) based on their expertise 
and gives an upper and lower quota for each class. 
Huang showed that it is NP-complete in general to decide the existence of a stable matching,
and proved that it is solvable in polynomial time 
if classes form a laminar family. 
For this tractable special case, Fleiner and Kamiyama \cite{FK16} 
gave a concise explanation in terms of matroids, and their framework is generalized by Yokoi \cite{Yokoi17}
to a framework with generalized matroids.

To cope with the nonexistence of a stable matching in constrained matching models
(not only models with lower quotas but also with other types of constraints such as regional constraints),
several relaxations of stability have been proposed.
See, e.g., Kamada and Kojima \cite{KK14, KK16}, Fragiadakis et al. \cite{Fragiadakis16}, and Goto et al. \cite{Goto16}.
Envy-freeness is one of them that places emphasis on fairness rather than efficiency.
Fragiadakis et al. \cite{Fragiadakis16} provided a strategy-proof algorithm that 
always finds an envy-free matching (or fair matching, in their terminology) of HR-LQ
under the assumption that all doctor-hospital pairs are acceptable.
The outcome of their mechanism also fulfills a second-best efficiency (i.e., nonwastefulness) property. 
Their framework is generalized in Goto et al. \cite{Goto16}
so that regional quotas can be handled.

Here we compare our models with the above models.
Unlike the models of Goto et al. \cite{Goto16} and Kamada and Kojima \cite{KK14, KK16}, our models cannot handle regional quotas.
Instead, our CSM model (in Sections \ref{sec:CSM} and \ref{sec:g-mat}) allows each hospital to have quotas on classes of doctors,
which are not dealt with in their models.
The setting of a tractable special case of CSM described in Section~\ref{sec:g-mat} 
is equivalent to a many-to-one case of Yokoi's model \cite{Yokoi17},
which studied stable matchings.
Neither \cite{Yokoi17} nor the study in this paper relies on the results of the other, while
both of them utilize the matroid framework of Fleiner \cite{Fleiner01, Fleiner03}.
%

\medskip
The remainder of this paper is organized as follows.
Section~\ref{sec:HR} investigates envy-free matchings in 
the Hospitals/Residents problem with lower quotas (HR-LQ).
In Section~\ref{sec:CSM}, 
we define an envy-free matching in the classified stable matching (CSM) model,
and show the NP-hardness of its existence test.
As its tractable special case, 
Section~\ref{sec:g-mat} presents results on CSM with paramodular quota functions.
Proofs for the theorems and corollary in Section~\ref{sec:g-mat}
are provided in the Appendix.

\section{Envy-freeness in HR with lower quotas}
\label{sec:HR}
In this section, we investigate envy-free matchings in 
the Hospitals/Residents problem with lower quotas (HR-LQ).

There are two disjoint sets $D$ and $H$, 
which represent doctors and hospitals, respectively.
A set of acceptable doctor-hospital pairs is denoted by $E\subseteq D\times H$.

\noindent
For each doctor $d\in D$, its acceptable hospital set is denoted by 
\[A(d):=\set{h\in H|(d,h)\in E}\subseteq H,\]  
and $d$ has a preference list (strict order) $\succ_{d}$ on $A(d)$.
Similarly, for each hospital $h\in H$, 
\[A(h):=\set{d\in D|(d,h)\in E}\subseteq D,\]  
and $h$ has a preference $\succ_{h}$ on $A(h)$.
Each hospital $h$ has
a lower quota $l_{h} \in \Z$ and an upper quota $u_{h}\in \Z$ with 
\[0\leq l_{h}\leq u_{h}\leq |A(h)|.\]

We call a tuple $I=(D, H, E, \succ_{DH}, \{(l_{h}, u_{h})\}_{h\in H})$ an {\bf HR-LQ instance},
where $\succ_{DH}$ is an abbreviated notation for the union of $\{\succ_{d}\}_{d\in D}$ and $\{\succ_{h}\}_{h\in H}$.
In particular, if $l_{h}=0$ for all $h\in H$, we call it an {\bf HR instance}.
An arbitrary subset  $M$ of $E$ is called an {\bf assignment}. 
For any assignment $M$, we denote
$M(d)=\set{h\in A(d)|(d,h)\in M}$ for each $d\in D$ and 
$M(h)=\set{d\in A(h)|(d,h)\in M}$ for each $h\in H$.
If $|M(d)|=1$, the notation $M(d)$ is also used to refer its single element. 

An assignment $M$ is called a {\bf matching} (or, said to be {\bf feasible}) if
$|M(d)|\leq 1$ for each $d\in D$ and $l_{h}\leq|M(h)|\leq u_{h}$ for each $h\in H$.
In a matching $M$, a doctor  $d$ is {\bf unassigned} (resp., {\bf assigned}) 
if $M(d)=\emptyset$ (resp., $|M(d)|=1$), and $h$ is {\bf undersubscribed} (resp., {\bf full})
if $|M(h)|<u_{h}$ (resp., $|M(h)|=u_{h}$).

\begin{definition}
	\label{def:stableHR}
	For a matching $M$, an unassigned pair $(d,h)\in E\setminus M$ is a {\bf blocking pair} if \\
	(i) $d$ is unassigned or $h\succ_{d} M(d)$, and
	(ii) $h$ is undersubscribed or there is $d'\in M(h)$ with $d\succ_{h} d'$.
	A matching $M$ is {\bf stable} if there is no blocking pair.
\end{definition}
For an HR instance, it is known that 
the algorithm of Gale and Shapley \cite{GS62}
always finds a stable matching.  The set of stable matchings has the following property.
\begin{proposition}[{\bf ``Rural Hospitals'' Theorem} \cite{GS85, Roth84b, Roth86}]
	\label{prop:ruralHR}
	For a given HR instance, 
	any two stable matchings $M, M'$ satisfy
	$|M(h)|=|M'(h)|$ for every $h\in H$. Moreover $M(h)=M'(h)$ if  $h$ is undersubscribed in $M$ or $M'$.
\end{proposition}
As mentioned in the Introduction, if hospitals have lower quotas,
then we cannot guarantee the existence of a stable matching anymore.
By Proposition~\ref{prop:ruralHR}, however, we can easily check the existence
by finding a stable matching while disregarding lower quotas, and 
checking whether the obtained matching meets lower quotas.

For an instance that has no stable matching, we want to obtain some matching that still has a kind of fairness.
As a relaxation of the concept of stability,
envy-freeness (also called fairness) of matchings has been proposed \cite{Fragiadakis16, WR16}.

\begin{definition}
	\label{def:envyfreeHR}
	For a matching $M$, a doctor $d$ has {\bf justified envy} toward $d'$ with $M(d')=h$ if 
	(i) $d$ is unassigned or $h\succ_{d} M(d)$ and 
	(ii) $d\succ_{h} d'$.
	A matching $M$ is {\bf envy-free} if no doctor has justified envy.
\end{definition}

Note that, if $d$ has justified envy toward $d'$ with $M(d)=h$,
then it means that $(d,h)$ is a blocking pair.
Thus, stability implies envy-freeness.
The envy-freeness of a matching is also regarded as the stability with reduced upper quotas, as follows.
\begin{proposition}\label{prop:HRtruncate}
	For $I=(D, H, E, \succ_{DH}, \{(l_{h}, u_{h})\}_{h\in H})$,
	an assignment $M$ is an envy-free matching 
	if and only if $M$ is a stable matching of $I'=(D, H, E, \succ_{DH}, \{(l_{h}, u'_{h})\}_{h\in H})$
	for some  $\{u'_{h}\}_{h\in H}$ with $u'_{h}\leq u_{h}~(h\in H)$.
\end{proposition}
\begin{proof}
	The ``if'' part is clear because feasibility in $I'$ implies that in $I$, and stability implies envy-freeness. 
	For the ``only if'' part, suppose that $M$ is envy-free in $I$
	and set $u'_{h}:=|M(h)|$ for each $h\in H$.
	Then,  $M$ is feasible for $I'$ and all hospitals are full,
	and hence there is no doctor who claims a hospital's vacant seat.
	Because $M$ is envy-free, it is stable in $I'$.
\end{proof}
By Proposition~\ref{prop:HRtruncate}, 
to check whether we can obtain a stable matching by reducing upper quotas,
it suffices to check for the existence of an envy-free matching.

Under the assumption that all doctor-hospital pairs are acceptable
and the sum of lower quotas does not exceed the number of doctors, 
Fragiadakis et al. \cite{Fragiadakis16} provided a strategy-proof mechanism that always finds an envy-free matching.
As a corollary, we have the following.
\begin{proposition}\label{prop:natural}
	For an instance $I=(D, H, E, \succ_{DH}, \{(l_{h}, u_{h})\}_{h\in H})$
	such that $E=D\times H$ and $|D|\geq \sum_{h\in H} l_{h}$,
	there exists an envy-free matching.
\end{proposition}
However, if not all pairs are acceptable, then even an envy-free matching may not exist.
Figure~1 shows an instance with $D=\{d_{1}, d_{2}\}$, $H=\{h_{1}, h_{2}\}$, $E=\{(d_{1}, h_{1}), (d_{2}, h_{1}), (d_{2}, h_{2})\}$,
$l_{h_{1}}=l_{h_{2}}=1$, and $u_{h_{1}}=u_{h_{2}}=2$.
For this instance, $M=\{(d_{1}, h_{1}), (d_{2}, h_{2})\}$
is the unique feasible matching,
but it is not envy-free because $d_{2}$ has justified envy toward $d_{1}$.
Hence, there is no envy-free matching. 

\begin{figure}[htb]
	\label{fig1}
	\begin{center}
		\includegraphics[width=0.6\hsize]{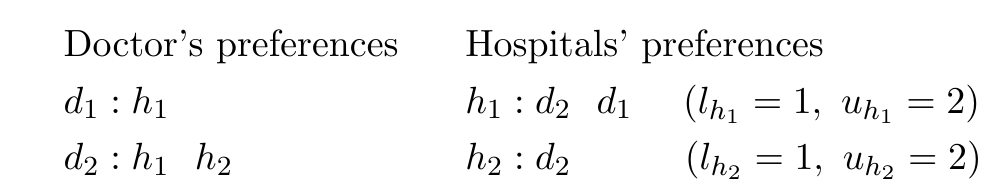}
		\caption{\small An instance of HR-LQ with no envy-free matching}
	\end{center}
\end{figure}

Note that an envy-free matching does exist if there is no lower quota, because empty matching is clearly envy-free.
Therefore, the existence test of an envy-free matching
is non-trivial when incomplete lists and lower quotas 
are introduced simultaneously.
Here we provide a characterization.

\begin{theorem}
	\label{thm:HR1}
	$I=(D, H, E, \succ_{DH}, \{(l_{h}, u_{h})\}_{h\in H})$
	has an envy-free matching if and only if 
	some stable matching $M'$ of the HR instance $I'=(D, H, E, \succ_{DH}, \{(0, l_{h})\}_{h\in H})$
	satisfies $|M'(h)|=l_{h}$ for all $h\in H$.
\end{theorem}
\begin{proof}
	For the ``if'' part, let $M'$ be a stable matching of $I'$ satisfying $|M'(h)|=l_{h}$ for all $h\in H$.
	Then, $M'$ is feasible for $I'$ and no doctor has justified envy because $M'$ has no blocking pair. 
	Thus, $M'$ is an envy-free matching of $I$.
	
	For the ``only if'' part, assume that $I$ has an envy-free matching $M$.
	Suppose, to the contrary, a stable matching $M'$ of $I'$ satisfies $|M'(h^{\ast})|<l_{h^{\ast}}$ 
	for some $h^{\ast}\in H$. 
	Let us denote $N=M\setminus M'$ and $N'=M'\setminus M$.
	For every $h\in H$, because  $|M'(h)|\leq l_{h}\leq |M(h)|$, we have $|N'(h)|\leq |N(h)|$. 
	In particular, $|N'(h^{\ast})|< |N(h^{\ast})|$ follows from $|M'(h^{\ast})|<l_{h^{\ast}}$.
	
	Consider a bipartite graph $G=(D, H; N\cup N')$,
	i.e., a graph between doctors and hospitals with edge set $N\cup N'=M\triangle M'$.
	Let $G^{\ast}$ be a connected component of $G$ including $h^{\ast}$, and denote
	by $D^{\ast}$ and $H^{\ast}$ the sets of doctors and hospitals in $G^{\ast}$, respectively. 
	Because there is no edge connecting $G^{\ast}$ and the outside,
	$\sum_{d\in D^{\ast}} |N(h)|=\sum_{h\in H^{\ast}} |N(h)|$ and 
	$\sum_{d\in D^{\ast}} |N'(h)|=\sum_{h\in H^{\ast}} |N'(h)|$.
	As $|N'(h^{\ast})|< |N(h^{\ast})|$ and $|N'(h)|\leq |N(h)|$ for any $h\in H^{\ast}$, we obtain 
	\[\textstyle{\sum_{d\in D^{\ast}} |N'(h)|=\sum_{h\in H^{\ast}} |N'(h)|<\sum_{h\in H^{\ast}} |N(h)|=\sum_{d\in D^{\ast}} |N(h)|}.\] 
	Then, there exists  $d^{\ast}\in D^{\ast}$ with $|N'(d^{\ast})|<|N(d^{\ast})|$,
	which implies $N'(d^{\ast})=\emptyset$ and $|N(d^{\ast})|=1$
	because $N'=M'\setminus M$ and $N=M\setminus M'$ are subsets of matchings.
	As $G^{\ast}$ is a connected bipartite graph,  
	there is a path $d_{0}~\!h_{0}~\!d_{1}~\!h_{1}\dots d_{k}~\!h_{k}$ with $d_{0}=d^{\ast}$ and $h_{k}=h^{\ast}$.
	Also, as $|N(d_{i})|\leq 1$ and $|N'(d_{i})|\leq 1$ for $i=0,1,\dots k$,
	this path alternately uses edges in $N=M\setminus M'$ and $N'=M'\setminus M$.
	Because $N'(d^{\ast})=\emptyset$ and $|N(d^{\ast})|=1$, we have 
	\begin{align*}
	M'(d_{0})&=\emptyset,\\
	(d_{i}, h_{i})&\in M\setminus M'\quad\quad(i=0,1,\dots,k),\\
	(d_{i+1},h_{i})&\in M'\setminus M\quad\quad(i=0,1,\dots,k-1).
	\end{align*}
	The doctor $d_{0}$ is unassigned in $M'$ and finds $h_{0}$ acceptable because $(d_{0}, h_{0})\in M$.
	Hence, the stability of $M'$  implies 
	that $h_{0}$ prefers $d_{1}\in M'(h_{0})$ to $d_{0}$.
	Then, the envy-freeness of $M$ implies that $d_{1}$ prefers 
	$h_{1}=M(d_{1})$ to $h_{0}$.
	In this way, we obtain 
	\begin{align*}
	&d_{i+1}\succ_{h_{i}} d_{i}\quad\quad(i=0,1,\dots,k-1),\\
	&h_{i+1}\succ_{d_{i+1}} h_{i}\quad\quad(i=0,1,\dots,k-1).
	\end{align*}
	Thus, $M(d_{k})=h_{k}\succ_{d_{k}} h_{k-1}=M'(d_{k})$.
	Because $h_{k}=h^{\ast}$ satisfies $|M'(h_{k})|<l_{h_{k}}$,
	then $(d_{k}, h_{k})$ is a blocking pair in $I'$,
	which contradicts the stability of $M'$.
\end{proof}
Theorem~\ref{thm:HR1} ensures that the following algorithm decides the existence of an envy-free matching of 
an HR-LQ instance $I=(D, H, E, \succ_{DH}, \{(l_{h}, u_{h})\}_{h\in H})$.

\vspace{3mm}
\noindent{\bf Algorithm} {\sf EF-HR-LQ}
\begin{description}
	\item[\sf Step1.] Find a stable matching $M'$ of $I'=(D, H, E, \succ_{DH}, \{(0, l_{h})\}_{h\in H})$. 
	\item[\sf Step2.] return $M'$ if $|M'(h)|=l_{h}$ for all $h\in H$, and otherwise ``there is no envy-free matching.''
\end{description}

Since the Gale-Shapley algorithm finds a stable matching
of an HR instance in $O(|E|)$ time, we obtain the following theorem.

\begin{theorem}\label{thm:HR2}
	For any HR-LQ instance $I=(D, H, E, \succ_{DH}, \{(l_{h}, u_{h})\}_{h\in H})$, 
	the algorithm {\sf EF-HR-LQ} 
	decides whether $I$ has an envy-free matching or not 
	in $O(|E|)$ time.
\end{theorem}

\section{Envy-freeness in Classified Stable Matching}
\label{sec:CSM}
In this section, we consider the envy-freeness in a model in which 
each hospital has lower and upper quotas on subsets of doctors.
This can be regarded as an envy-free matching version of the Classified Stable Matching, 
proposed by Huang \cite{CCH10}.
Similarly to Section~\ref{sec:HR}, we have
doctors $D$, hospitals $H$, acceptable pairs $E\subseteq D\times H$,
and preferences $\succ_{DH}$.

The only difference from HR-LQ is that, 
in the current model, each hospital $h\in H$ 
has a pair of functions $p_{h}, q_{h}:2^{A(h)}\to \Z$,
instead of a pair of numbers $l_{h}, u_{h}$.  
These functions define a lower and an upper quota for each subset of acceptable doctors.
Throughout this paper, we assume that for any hospital $h$, 
the functions $p_{h}$ and $q_{h}$ satisfy
\[0\leq p_{h}(B)\leq q_{h}(B)\leq |B|\qquad (B\subseteq A(h)).\] 
We call such a tuple $(D, H, E, \succ_{DH}, \{(p_{h}, q_{h})\}_{h\in H})$ a {\bf CSM instance}.
For each $h\in H$, the family of {\bf acceptable} subsets of doctors is denoted by
\[\F(p_{h}, q_{h}):=
\set{X\subseteq A(h)| \forall B\subseteq A(h):p_{h}(B)\leq |X\cap B|\leq q_{h}(B)}.\]
For any $h\in H$, we say that $B\subseteq A(h)$ has 
a {\bf non-trivial lower} (resp., {\bf upper}) {\bf constraint}
if $p_{h}(B)>0$ (resp., $q_{h}(B)<|B|$).
We denote the family of constrained subsets by 
\[\C(p_{h}, q_{h}):=\set{B\subseteq A(h)|p_{h}(B)>0 \text{~ or ~} q_{h}(B)<|B|}.\]
Then, we see that $\F(p_{h}, q_{h})$ is represented as
\[\F(p_{h}, q_{h})=\set{X\subseteq A(h)| \forall B\subseteq \C(p_{h}, q_{h}):p_{h}(B)\leq |X\cap B|\leq q_{h}(B)}.\]

For a CSM instance $I=(D, H, E, \succ_{DH}, \{(p_{h}, q_{h})\}_{h\in H})$, 
$M\subseteq E$ is called a {\bf matching} (or, said to be {\bf feasible}) if
$|M(d)|\leq 1$ for each $d\in D$ and $M(h)\in \F(p_{h}, q_{h})$ for each $h\in H$. 

\begin{definition}\label{def:stableCSM}
	For a matching $M$, an unassigned pair $(d,h)\in E\setminus M$ is a {\bf blocking pair} if 
	(i) $d$ is unassigned or $h\succ_{d} M(d)$, and
	(ii) $M(h)+d\in \F(p_{h}, q_{h})$ or 
	$M(h)+d-d'\in \F(p_{h}, q_{h})$ for some $d'\in M(h)$ \mbox{with $d\succ_{h} d'$}.
	A matching $M$ is {\bf stable} if there is no blocking pair.
\end{definition}
In Definition~\ref{def:stableCSM}, 
the condition $M(h)+d\in \F(p_{h}, q_{h})$ means that 
$h$ can add $d$ to the current assignment without violating any upper quota, 
and $M(h)+d-d'\in \F(p_{h}, q_{h})$ means that
$h$ can replace $d'$ with $d$ without violating any upper or lower quota.
The Classified Stable Matching, introduced by Huang \cite{CCH10},
is the problem to decide the existence of a stable matching for a given CSM instance%
\footnote{In his original model, each hospital $h$ has a classification $\C_{h}\subseteq 2^{A(h)}$ and 
	sets a lower and an upper quota for each member of $\C_{h}$. 
	That is, we are provided $\C(p_{h}, q_{h})$ and the values of $p_{h}$, $q_{h}$ on it, 
	rather than set functions $p_{h}$, $q_{h}$.
	Our formulation uses set functions to simplify the arguments in the next section.}%
.
Because this is a generalization of HR-LQ, there are instances that have no stable matching.
Let us consider envy-freeness for a CSM instance.
\begin{definition}
	For a matching $M$, a doctor $d$ has {\bf justified envy} toward $d'$ with $M(d')=h$ if 
	(i) $d$ is unassigned or $h\succ_{d} M(d)$ and 
	(ii) $M(h)+d-d'\in \F(p_{h}, q_{h})$ and $d\succ_{h} d'$.
	A matching $M$ is {\bf envy-free} if no doctor has justified envy.
\end{definition}

As in the case of HR-LQ,
an envy-free matching can be regarded as a stable matching 
with reduced upper quotas as follows.
For any $h\in H$ and $k\in \Z$ with $0\leq k\leq q(A(h))$,
a function $q'_{h}:2^{A(h)}\to \Z$ is called a {\bf $k$-truncation} of $q_{h}$ if 
$q'(A(h))=k$ and $q'(B)=q(B)$ for every $B\subsetneq A(h)$.
Also, we simply say that $q'_{h}$ is a {\bf truncation} of $q_{h}$ if 
there is such $k\in \Z$.

\begin{proposition}
	\label{prop:CSMtruncation}
	For  $I=(D, H, E, \succ_{DH}, \{(p_{h}, q_{h})\}_{h\in H})$,
	an assignment $M$ is an envy-free matching 
	if and only if $M$ is a stable matching of
	$I'=(D, H, E, \succ_{DH}, \{(p_{h}, q'_{h})\}_{h\in H})$
	such that each $q'_{h}$ is some truncation of $q_{h}$.
\end{proposition}
\begin{proof}
	To show the ``only if'' part, let $M$ be an envy-free matching of $I$.
	For each $h\in H$, let $q'_{h}$ be the $|M(h)|$-truncation of $q_{h}$.
	Then $M(h)\in \F(p_{h}, q'_{h})$ and $M(h)+d\not\in \F(p_{h}, q'_{h})$
	for every $d\in A(h)\setminus M(h)$.
	That is, $M$ is feasible for $I'$ and there is no doctor who claims a hospital's vacant seat. 
	Therefore, if there is a blocking pair $(d,h)\in E\setminus M$ for $I'$,
	it follows that 
	$d$ has a justified envy toward some $d'$ with $M(d')=h$,
	which contradicts the envy-freeness of $M$.
	Thus, $M$ is a stable matching of $I'$.
	
	For the ``if'' part, let $M$ be a stable matching of $I'$.
	Clearly, $M$ is feasible for $I$.
	Suppose, to the contrary,
	some doctor $d$ has justified envy toward $d'$ with $M(d')=h$ with respect to $I$.
	Then $d$ is unassigned or $h\succ_{d}M(d)$. 
	Also, we have $d\succ_{h}d'$ and $M(h)+d-d'\in \F(p_{h}, q_{h})$.
	Then, $M(h)+d-d'\in \F(p_{h}, q'_{h})$ follows because $|M(h)+d-d'|=|M(h)|$. 
	Hence, $(d,h)$ is a blocking pair in $I'$, a contradiction.
\end{proof}

We provide a hardness result for deciding the existence of an envy-free matching.
Here, we assume that evaluation oracles of 
set functions $p_{h}$ and $q_{h}$ are available for each hospital $h$.
\begin{theorem}
	It is NP-hard to decide whether a CSM instance
	$I=(D, H, E, \succ_{DH}, \{(p_{h}, q_{h})\}_{h\in H})$
	has an envy-free matching or not.
	The problem is NP-complete even if the size of $\C(p_{h}, q_{h})$ is at most 4 for each $h\in H$.
\end{theorem}
\begin{proof}
We use reduction from the NP-complete problem (3, B2)-SAT \cite{BKS03}, which is a restriction of SAT
such that each clause contains
exactly three literals and each variable occurs exactly twice as a positive literal and exactly twice as a negative literal.
Let $\varphi=c_{1}\land c_{2}\land\cdots\land c_{m}$ 
be an instance of (3, B2)-SAT with 
Boolean variables $v_{1}, v_{2},\dots,v_{n}$.
Then, each clause $c_{j}$ is a disjunction of three literals, 
(e.g., $c_{j}=v_{1}\lor \lnot v_{2}\lor \lnot v_{3}$)
and each of literals $v_{i}$ and $\lnot v_{i}$ appears in exactly two clauses.
For each variable $v_{i}$, denote by $j^{\ast}(i,1)$, $j^{\ast}(i,2)$ 
the indices of two clauses that contain $v_{i}$.
Similarly, denote by $j^{\ast}(i,-1)$, $j^{\ast}(i,-2)$ 
the indices of clauses that contain $\lnot v_{i}$.

We now define a CSM instance corresponding to $\varphi$.
We have a variable-hospital $h_{i}$ for each variable $v_{i}$,
and a clause-hospital $h_{j}$ for each clause $c_{j}$.
For each variable $v_{i}$,
we have four doctors $\set{d_{i,t}|t\in \{1,2,-1,-2\}}$.
For each doctor $d_{i,t}$, we have
\[A(d_{i,t})=\{h_{i},h_{j^{\ast}(i,t)}\},\quad h_{i}\succ_{d_{i,t}} h_{j^{\ast}(i,t)}.\]
The set $E$ is defined as the set of all pairs $(d_{i,t},h)$ such that $h\in A(d_{i,t})$.
Then, for each variable-hospital $h_{i}$ and clause-hospital $h_{j}$, we have  
\begin{align*}
&A(h_{i})=\set{d_{i,t}|t\in \{1,2,-1,-2\}},\\
&A(h_{j})=\set{d_{i,t}|j^{\ast}(i,t)=j}.
\end{align*}
Note that $d_{i,t}\in A(h_{j})$ implies $v_{i}\in c_{j}$ or $\lnot v_{i}\in c_{j}$.
Also, each of $v_{i}\in c_{j}$ and $\lnot v_{i}\in c_{j}$ implies $d_{i,t}\in A(h_{j})$
for some unique $t\in\{1,2,-1,-2\}$.
Therefore,  $|A(h_{j})|=3$ for each clause-hospital $h_{j}$.
For each variable-hospital $h_{i}$, define $p_{h_{i}}$ and $q_{h_{i}}$ so that 
\begin{align*}
&\C(p_{h_{i}}, q_{h_{i}})=\textstyle{\bigcup}\set{\{d_{i,t}, d_{i,t'}\}| t\in \{1,2\},~t'\in \{-1,-2\}},\\
&p_{h_{i}}(\{d_{i,t}, d_{i,t'}\})=q_{h_{i}}(\{d_{i,t}, d_{i,t'}\})=1\qquad (t\in \{1,2\},~t'\in \{-1,-2\}).
\end{align*}
Then,  we see that
$\F(p_{h_{i}}, q_{h_{i}})=\{D^{+}_{i}, D^{-}_{i}\}$, where $D^{+}_{i}:=\{d_{i,1}, d_{i,2}\}$ and $D^{-}_{i}:=\{d_{i,-1}, d_{i,-2}\}$.
For each clause-hospital $h_{j}$, define $p_{h_{i}}$ and $q_{h_{i}}$ so that 
\[\C(p_{h_{j}}, q_{h_{j}})=\{A(h_{j})\},~~
p_{h_{j}}(A(h_{j}))=1,~~q_{h_{j}}(A(h_{j}))=|A(h_{j})|=3.\]
We define preference lists of hospitals arbitrarily.
Note that $|\C(p_{h}, q_{h})|\leq 4$ for every hospital.
We show that this CSM instance has an envy-free matching
if and only if  $\varphi=c_{1}\land c_{2}\land\cdots\land c_{m}$ is satisfiable.

\smallskip

{\bf The ``only if'' part:}
Suppose that there is an envy-free matching $M$.
Then, for every variable-hospital $h_{i}$, 
$M(h_{i})$ is $D^{+}_{i}$ or $D^{-}_{i}$.
For each $h_{i}$, set variable $v_{i}$ to {\sf FALSE} if $M(h_{i})=D^{+}_{i}$, 
and to {\sf TRUE} if $M(h_{i})=D^{-}_{i}$.
This Boolean assignment satisfies every clause $c_{j}$ of $\varphi$ as follows.
Because $M(h_{j})\in\F(p_{h_{j}}, q_{h_{j}})$,
we have $|M(h_{j})|\geq 1$. Hence,
some $d_{i,t}$ with $j^{\ast}(i,t)=j$ is assigned to $h_{j}$.
Then, $d_{i,t}\not\in M(h_{i})$.
There are two cases: (i) $t\in \{1,2\}$, (ii) $t\in \{-1,-2\}$.
In the case (i), $d_{i,t}\not\in M(h_{i})$ implies $M(h_{i})\neq D^{+}_{i}$, and hence 
$v_{i}$ is set to {\sf TRUE}. Also, $t\in \{1,2\}$ and $j^{\ast}(i,t)=j$ imply
$v_{i}\in c_{j}$. Hence, clause $c_{j}$ is satisfied.
Similarly, in the case (ii), we see that
$v_{i}$ is set to {\sf FALSE} and we have $\lnot v_{j}\in c_{j}$. 
Hence, clause $c_{j}$ is satisfied.
\smallskip

{\bf The ``if'' part:}
Suppose that there is a Boolean assignment satisfying $\varphi$.
Define an assignment $M$ so that
\begin{itemize}
	\setlength{\itemsep}{0mm}
	\item 
	$M(h_{i})=D^{-}_{i}$ if $v_{i}$ is {\sf TRUE}, and $M(h_{i})=D^{+}_{i}$ if $v_{i}$ is {\sf FALSE}, and
	\item  
$M(h_{j})=
\set{d_{i,t}\in A(h_{j})|d_{i,t}\in D^{+}_{i}, ~\text{$v_{i}$ is {\sf TRUE}}}\cup 
\set{d_{i,t}\in A(h_{j})|d_{i,t}\in D^{-}_{i}, ~\text{$v_{i}$ is {\sf FALSE}}}$.
\end{itemize}
We can observe that $|M(d)|=1$ for every doctor $d$,
and $M(h_{i})\in \F(p_{h_{i}},q_{h_{i}})$ for every variable-hospital $h_{i}$.
Also, because all clauses are satisfied,  
the above definition implies 
$M(h_{j})\in \F(p_{h_{j}},q_{h_{j}})$ for every clause-hospital $h_{j}$.
Then, $M$ is feasible. We now show the envy-freeness of $M$.
Suppose, to the contrary, $d_{i,t}$ has justified envy toward $d'$.
Because we have $|M(d_{i,t})|=1$,  $A(d_{i,t})=\{h_{i}, h_{j^{\ast}(i,t)}\}$,  
and $h_{i}\succ_{d_{i,t}} h_{j^{\ast}(i,t)}$,
this justified envy implies conditions $d'\in M(h_{i})$, $d_{i,t}\not\in M(h_{i})$ and 
$M(h_{i})+d_{i,t}-d'\in \F(p_{h_{i}}, q_{h_{i}})$.
As $M(h_{i})\in \F(p_{h_{i}}, q_{h_{i}})=\{D^{+}_{i}, D^{-}_{i}\}$,
then we have
$\{M(h_{i})+d_{i,t}-d', M(h_{i})\}=\{D^{+}_{i}, D^{-}_{i}\}$,
which contradicts 
$|D^{+}_{i}\setminus D^{-}_{i}|=|D^{-}_{i}\setminus D^{+}_{i}|=2$.
\end{proof}

\section{Envy-freeness in CSM with Paramodular Quotas} 
\label{sec:g-mat}
In Section~\ref{sec:CSM}, we showed that it is NP-hard in general to
decide whether a CSM instance has an envy-free matching or not.
This section shows that the problem is solvable in polynomial time 
if the pair of quota functions is paramodular for each hospital. 
The proofs of the theorems and corollary in this section are provided in the Appendix. 
We first introduce the notion of paramodularity \cite{Frankbook}.

Let $A$ be a finite set and let $p, q:2^{A}\to \Z$.
The pair $(p,q)$ is {\bf paramodular} (or, called a {\bf strong pair} \cite{FT88}) if
\begin{itemize}
	\setlength{\itemsep}{-0.2mm}
	\item $p$ is {\bf supermodular}, i.e., $p(B)+p(B')\leq p(B\cup B')+p(B\cap B')$ for every $B, B'\subseteq A$,
	\item $q$ is {\bf submodular}, i.e., $q(B)+q(B')\geq q(B\cup B')+q(B\cap B')$ for every $B, B'\subseteq A$, and 
	\item 
	the {\bf cross-inequality}
	$q(B)-p(B')\geq q(B\setminus B')-p(B'\setminus B)$ holds for every $B, B'\subseteq A$.
\end{itemize}

Here we provide examples of constraints that can be represented by 
paramodular pairs.
(See Yokoi \cite[Appendices A and B]{Yokoi17}.)

\begin{example}[Laminar Constraints]\label{ex:1}
Let 
$\laminar\subseteq 2^{A}$ be a laminar (or hierarchical) classification (i.e.,
any $X, Y\subseteq \laminar$ satisfy $X\subseteq Y$ or $X\supseteq Y$ or $X\cap Y\neq \emptyset$).
Let $\hat{p}, \hat{q}: \laminar\to \Z$ be functions that define a lower and an upper quota for each class.
Denote the acceptable set family by
$\J(\laminar, \hat{p}, \hat{q}):=\set{B\subseteq A| \forall X\in \laminar: \hat{p}(X)\leq |B\cap X|\leq \hat{q}(X)}$.
If $\J(\laminar, \hat{p}, \hat{q})$ is nonempty, then $\J(\laminar, \hat{p}, \hat{q})=\F(p,q)$
for some paramodular pair $(p,q)$.
\end{example}

\begin{example}[Staffing Constraints]\label{ex:2}
For a finite set 
$S$ (e.g., a set of sections of a hospital), let $\Gamma:S\to 2^{A}$ and $\hat{l}, \hat{u}:S\to \Z$ be functions such that 
$\Gamma(s)\subseteq A$ represents members acceptable to $s\in S$
and $\hat{l}(s), \hat{u}(s)\in \Z$ represent a lower and an upper quota of each $s\in S$.
Let $\J(S,\Gamma,\hat{l}, \hat{u})\subseteq 2^{A}$ be a family of subsets $B\subseteq A$ such that there exists 
a function $\pi:B\to S$ satisfying $\forall d\in B: d\in \Gamma(\pi(d))$ and
$\forall s\in S: \hat{l}(s)\leq |\set{d\in B|\pi(d)=s}|\leq \hat{u}(s)$.
If $\J(S,\Gamma,\hat{l}, \hat{u})$ is nonempty, then $\J(S,\Gamma,\hat{l},\hat{u})=\F(p,q)$ for some paramodular pair $(p,q)$.
\end{example}

For a set function $p:2^{A}\to \Z$, its {\bf complement} $\overline{p}:2^{A}\to \Z$ 
is defined by
\[\overline{p}(B)=p(A)-p(A\setminus B)\quad (B\subseteq A).\]

Recall that a CSM instance is represented as a tuple
$(D, H, E, \succ_{DH}, \{(p_{h}, q_{h})\}_{h\in H})$,
where it is assumed that $0\leq p_{h}(B)\leq q_{h}(B)\leq |B|$ for every $h\in H$ and $B\subseteq A(h)$.
Here is the main theorem of this section.
We denote by $\zero$ a set function that is identically zero.

\begin{theorem}\label{thm:g-mat1}
	For a CSM instance $I=(D, H, E, \succ_{DH}, \{(p_{h}, q_{h})\}_{h\in H})$,
	suppose that $(p_{h}, q_{h})$ is paramodular for each $h\in H$.
	Then, an instance $I'=(D, H, E, \succ_{DH}, \{(\zero,\overline{p_{h}})\}_{h\in H})$
	has at least one stable matching and the following three conditions are equivalent.
	\begin{enumerate}
		\setlength{\leftskip}{0.5cm}
		\item[\rm (a)] $I$ has an envy-free matching.
		\item[\rm (b)] Some stable matching $M'$ of $I'$ satisfies $|M'(h)|=p_{h}(A(h))$ for all $h\in H$.
		\item[\rm (c)] Every stable matching $M'$ of $I'$ satisfies $|M'(h)|=p_{h}(A(h))$ for all $h\in H$.
	\end{enumerate}
	Also, if {\rm (b)} holds, then $M'$ is an envy-free matching of $I$.
\end{theorem}
As will be shown in Appendix~\ref{sec:proof4}, the existence of a stable matching of $I'$ and
the equivalence between (b) and (c) follows from Fleiner's results on the matroid framework \cite{Fleiner01, Fleiner03}.
The most difficult part is showing the equivalence between conditions (a) and (b).
To show that (a) implies (b), we construct a stable matching $M'$ of $I'$ from an envy-free matching $M$ of $I$. 
This construction is achieved by using the fixed-point method of Fleiner \cite{Fleiner03}. 
The paramodularity of each $(p_{h}, q_{h})$
(or a generalized matroid structure of each $\F(p_{h}, q_{h})$) is essential
to show the existence of a fixed-point satisfying a required condition
(see Lemma~\ref{lem:1and2} in Appendix~\ref{sec:proof4} for the details).

By Theorem~\ref{thm:g-mat1}, when quota function pairs are paramodular,
we can decide the existence of an envy-free matching of $I=(D, H, E, \succ_{DH}, \{(p_{h}, q_{h})\}_{h\in H})$
by the following algorithm. 
\begin{description}
	\setlength{\itemsep}{-1mm}
	\item[\sf Step1.] Find a stable matching $M'$ of $I'=(D, H, E, \succ_{DH}, \{(\zero, \overline{p_{h}})\}_{h\in H})$. 
	\item[\sf Step2.] If $|M'(h)|=p_{h}(A(h))$ for every $h\in H$, then return $M'$. 
	Otherwise, return ``there is no envy-free matching.''
\end{description}

As will be shown in the Appendix, Step 1 (i.e., finding a stable matching of $I'$)
can be done efficiently by the generalized Gale-Shapley algorithm studied in \cite{Fleiner01, Fleiner03}. 
The detailed description of the algorithm is as follows.
Here, for each $h\in H$, $N\subseteq E$, and $d\in N(h)$, we use the notation 
$N(h)_{\succ_{h} d}:=\set{d'\in N(h)|d'\succ_{h} d}$
and $N(h)_{\succeq_{h} d}:=\set{d'\in N(h)|d'\succ_{h} d\text{ ~or~ } d'=d}$.

\begin{algorithm}[htb]
	\SetAlgoLined                 
	\caption{\sf EF-Paramodular-CSM}         
	\label{alg1}
	{\bf Input:}  $I=(D, H, E, \succ_{DH}, \{(p_{h}, q_{h})\}_{h\in H})$ such that each $(p_{h}, q_{h})$ is paramodular
	{\bf Output:} return an envy-free matching $M'$, or ``there is no envy-free matching.''\\
	\vspace{2mm}                          
	Set $N_{D}\leftarrow E$, $N_{H}\leftarrow \emptyset$, and let $M'$ be undefined\; 
	\While{$M'$ is undefined}{        
		$R_{D}\leftarrow \bigcup_{d\in D} \set{(d,h) |h\in N_{D}(d),~h\neq \max_{\succ_{d}}N_{D}(d)}$\;	
		$R_{H}\leftarrow \bigcup_{h\in H} 
		\set{(d,h)| d\in N_{H}(h),~ p(A(h)\setminus N_{H}(h)_{\succeq_{h} d})=p(A(h)\setminus N_{H}(h)_{\succ_{h} d})}$\;		
		\eIf{$(N_{D}, N_{H})=(E\setminus R_{H}, E\setminus R_{D})$}{
			let $M'\leftarrow N_{D}\cap N_{H}$ and {\bf break}\;
		}{
			update $(N_{D}, N_{H})\leftarrow (E\setminus R_{H}, E\setminus R_{D})$\;
	}}
	\eIf{$|M'(h)|=p_{h}(A(h))$ for all $h\in H$}{
		return $M'$\;
	}{	
		return ``there is no envy-free matching'';
	}
\end{algorithm}
\vspace{4mm}
In the Appendix, we show that the assignment $M'$ obtained in the algorithm is
indeed a stable matching of $I'$. 
Also, it will be shown that $N_{D}$ is monotone decreasing and 
$N_{H}$ is monotone increasing in the algorithm, and hence the ``while loop'' is iterated at most 
$2|E|$ times.
Thus, we obtain the following theorem.
(See Appendix~\ref{sec:proof5} for the details.)

\begin{theorem}\label{thm:g-mat2}
	For a CSM instance $I=(D, H, E, \succ_{DH}, \{(p_{h}, q_{h})\}_{h\in H})$
	such that each $(p_{h}, q_{h})$ is paramodular,
	the algorithm {\sf EF-Paramodular-CSM} decides whether $I$ has an envy-free matching or not 
	in $O(|E|^{2})$ time, 
	provided that evaluation oracles of $\{p_{h}\}_{h\in H}$ are available.
\end{theorem}

As is shown in Examples~\ref{ex:1} and \ref{ex:2}, 
when the family of acceptable doctor sets of each hospital $h\in H$
is defined by a laminar constraint $\J_{h}:=\J(\laminar_{h}, \hat{p}_{h}, \hat{q}_{h})$ or 
by a staffing constraint $\J_{h}:=\J(S_{h}, \Gamma_{h}, \hat{l}_{h}, \hat{u}_{h})$,
then there is a paramodular pair $(p_{h}, q_{h})$ such that $\J_{h}=\F(p_{h}, q_{h})$.
The following corollary states that, 
in such a case, we can decide the existence of 
an envy-free matching of $I=(D, H, E, \succ_{DH}, \{(p_{h}, q_{h})\}_{h\in H})$
even if evaluation oracles of  $\{p_{h}\}_{h\in H}$ are not provided.

\begin{corollary}\label{cor:g-mat4}
Suppose that, for each $h\in H$, the family of acceptable doctor sets
is defined in the form $\J_{h}:=\J(\laminar_{h}, \hat{p}_{h}, \hat{q}_{h})\neq \emptyset$ 
$($resp., $\J_{h}:=\J(S_{h}, \Gamma_{h}, \hat{l}_{h}, \hat{u}_{h})\neq \emptyset$$)$.
Let $(p_h, q_{h})$ be a paramodular pair such that $\J_{h}=\F(p_h, q_{h})$.
Then, given $\laminar_{h}, \hat{p}_{h}, \hat{q}_{h}$ $($resp., $S_{h}, \Gamma_{h}, \hat{l}_{h}, \hat{u}_{h}$$)$ 
for each $h\in H$, one can decide whether $I=(D, H, E, \succ_{DH}, \{(p_{h}, q_{h})\}_{h\in H})$ 
has an envy-free matching or not in time polynomial in $|E|$ $($resp., in $|E|$ and $\max_{h\in H}|S_{h}|$$)$.
\end{corollary}

\begin{proof}
Since we have Theorem~\ref{thm:g-mat2},
it completes the proof to show that we can simulate an evaluation oracle of each $p_{h}$ in time polynomial in $|E|$
(resp., in $|E|$ and $|S_{h}|$).
By Proposition~\ref{prop:one-to-one} in Appendix~\ref{sec:proof1}, 
for each $B\subseteq A(h)$, the value of $p_{h}(B)$ is obtained as
$p_{h}(B)=\min\{~|X\cap B|\mid X\in \J_{h}\}$.
Consider a weight function $w_{B}$ on $A(h)$ such that
$w_{B}(d)=1$ for every $d\in B$ and $w_{B}(d)=0$ for every $d\in A(h)\setminus B$.
Then, $p_{h}(B)$ is written as 
$p_{h}(B)=\min\set{w_{B}(X)\mid X\in \J_{h}}$,
which is a weight minimization problem on a generalized matroid.
As explained in \cite[Appendix C]{Yokoi17}, when $\J_{h}$ is given in the form above,
this can be reduced to the minimum cost circulation problem,
which can be solved in strongly polynomial time \cite{Tardos85-2, Orlin93}.
(See \cite{Yokoi17} for the details of the reduction.) 
Thus, the proof is completed.
\end{proof}

\begin{remark}
	Theorems~\ref{thm:g-mat1} and \ref{thm:g-mat2}
	generalize Theorems~\ref{thm:HR1} and \ref{thm:HR2} as follows.
	For a pair $(l_{h}, u_{h})$ of nonnegative integers
	with $0\leq u_{h}\leq l_{h}\leq |A(h)|$, define
	$p_{h}, q_{h}:2^{A(h)}\to \Z$ by
	\begin{align*}
	&p_{h}(B)=\max\{0, l_{h}-|A(h)\setminus B|\},\quad q_{h}(B)=\min\{u_{h}, |B|\},\qquad(B\subseteq A(h)).
	\end{align*}
	Then, $(p_{h}, q_{h})$ is paramodular and $\F(p_{h}, q_{h})=\set{X\subseteq A(h)|l_{h}\leq |X|\leq u_{h}}$.
	Hence, envy-freeness 
	for $(D, H, E, \succ_{DH}, \{l_{h}, u_{h}\}_{h\in H})$ 
	coincides with that for $(D, H, E, \succ_{DH}, \{p_{h}, q_{h}\}_{h\in H})$.
	Also, we can check  $p_{h}(A(h))= \max\{0, l_{h}-|A(h)\setminus A(h)|\}=l_{h}$.
\end{remark}

\section{Acknowledgments}
I wish to thank the anonymous reviewers whose
comments have benefited the paper greatly.
This work was supported by JST CREST, Grant Number JPMJCR14D2, Japan.


\appendix
\section{Proofs for Section \ref{sec:g-mat}}
\label{sec:proof}
Here, we provide proofs of
Theorems~\ref{thm:g-mat1}, \ref{thm:g-mat2}.
This section consists of five subsections, and the first three introduce notions and previous results needed for the proofs.
More precisely, Sections~\ref{sec:proof1}, \ref{sec:proof2}, and \ref{sec:proof3} respectively 
introduce notions of generalized matroids, choice functions induced from matroids, and the lattice fixed-point method for stable matchings.
Using them, the last two subsections provide the proofs of our results.

\subsection{Generalized Matroids}
\label{sec:proof1}
For a finite set $A$ and a family $\J\subseteq 2^{A}$, 
the pair $(A,\J)$ is called a {\bf generalized matroid} \cite{Tardos85} ({\bf \mbox{g-matroid}}, for short) if $\J$ is nonempty and satisfies 
the following property called {\bf simultaneous} (or {\bf symmetric}) {\bf exchange property}%
\footnote{This is not the original definition of generalized matroids by Tardos \cite{Tardos85}, but equivalent to
it as shown by Murota and Shioura \cite{MS99}.} \cite{MS99}.

\begin{description}
	\item[(B$^{\natural}$-EXC)]
	For any $X, Y\in \J$ and $e\in X\setminus Y$, we have \\
	(i) $X-e\in \J$, $Y+e\in \J$ or\\
	(ii) there exists some $e'\in Y\setminus X$ such that $X-e+e'\in \J$, $Y+e-e'\in \J$.
\end{description} 

\noindent
The family  $\J$ of a g-matroid $(A, \J)$ is also called an {\bf M$^{\natural}$-convex family} \cite{Mbook, Murota16survey}.
(There are various characterizations for g-matroids.
See, e.g., Tardos \cite{Tardos85}, Frank \cite{Frankbook} and Murota \cite{Mbook} for more information on g-matroid and its extensions.)

For set functions $p,q:2^{A}\to \Z$, the pair $(p, q)$ is called {\bf g-matroidal} if it is paramodular and
satisfies $0\leq p(B)\leq q(B)\leq |B|$ for every $B\subseteq A$.
As its name indicates, there is a one-to-one correspondence between
generalized matroids and g-matroidal pairs (see, e.g., \cite{Frankbook, FT88}).
\begin{proposition}
	\label{prop:one-to-one}
	A pair $(A,\J)$ is a g-matroid if and only if $\J=\F(p,q)$ 
	for some g-matroidal pair $(p,q)$.
	Such a g-matroidal pair is uniquely defined by 
\begin{align*}
p(B)=\min\{~|X\cap B|\mid X\in \J\}\qquad(B\subseteq A),\\
q(B)=\max\{~|X\cap B|\mid X\in \J\}\qquad (B\subseteq A).
\end{align*}
\end{proposition}
By Proposition~\ref{prop:one-to-one}, 
the families $\J(\laminar, \hat{p}, \hat{q})$ and $\J(S,\Gamma,\hat{l}, \hat{u})$ 
defined in Examples~\ref{ex:1} and \ref{ex:2} are the independent set families of g-matroids.
(See Yokoi \cite[Appendices A and B]{Yokoi17} for examples and operations of g-matroids.) 

A function $r:2^{A}\to \Z$ is called a {\bf matroid rank function} if 
it is submodular, monotone (i.e., $B\subseteq B'\subseteq A$ implies $r(B)\leq r(B')$), 
and satisfies $0\leq r(B)\leq |B|$ for every $B\subseteq A$.
The submodularity of $r$ is equivalent to the following 
{\bf diminishing returns property}:
for any $B'\subseteq B\subseteq A$ and $e\in A\setminus B$, we have
\[r(B+e)-r(B)\leq r(B'+e)-r(B').\]
In particular,  a matroid rank function satisfies 
$0\leq r(B+e)-r(B)\leq r(\{e\})-r(\emptyset)\leq 1$ 
for any $B\subseteq A$ and $e\in A\setminus B$.

A pair $(A,\I)$ is called a {\bf matroid}
if it is a g-matroid and $\emptyset \in \I$.
In terms of quota functions,
a pair $(A,\I)$ is a matroid if there is a matroid rank function $r$ such that $\I=\F(\zero,r)$.
Indeed, we can check that the pair $(\zero, r)$ is g-matroidal
for any matroid rank function $r$.

\subsection{Choice Functions Induced from Matroid Rank Functions}
\label{sec:proof2}
Let $r:2^{A}\to \Z$ be a matroid rank function on $A$
and $\succ$ be a linear order on $A$.
Let $\M=(A, r,\succ)$ and define a function $\Cm:2^{A}\to 2^{A}$ as follows.
Let $n=|A|$ and, for $i=1,2,\dots, n$, let $e_{i}$ be the $i$-th best element of $A$ with respect to $\succ$, i.e.
$e_{1}\succ e_{2}\succ \cdots \succ e_{n}$.
Let $A_{0}:=\emptyset$ and $A_{i}:=\{e_{1}, e_{2},\dots,e_{i}\}$ for each $i=1,2, \dots ,n$.
Then, define $\Cm$  by
\[\Cm(X):=\set{e_{i}\in X| r(A_{i}\cap X)>r(A_{i-1}\cap X)}\quad (X\subseteq A).\]
We call $\Cm$ the {\bf choice function induced from} $\M=(A, r,\succ)$.
Note that, for any $e_{i}\in A$ and $X\subseteq A$, 
the value of $r(A_{i}\cap X)-r(A_{i-1}\cap X)$ is $1$ or $0$ by the monotonicity and submodularity of $r$.
Also, $e_{i}\in A\setminus X$ implies $r(A_{i}\cap X)-r(A_{i-1}\cap X)=0$.
Then, for any $X\subseteq A$, we have
\begin{eqnarray}
\Cm(X)=\set{e_{i}\in A| r(A_{i}\cap X)-r(A_{i-1}\cap X)=1},\label{eq:choice1}\\
X\setminus \Cm(X)=\set{e_{i}\in X| r(A_{i}\cap X)-r(A_{i-1}\cap X)=0}.
\label{eq:choice2}
\end{eqnarray}
Such a choice function was introduced by Fleiner \cite{Fleiner01, Fleiner03} and used in 
several works \cite{BFIM10, FK16, Yokoi17}.
In these works,  matroids are usually given by independent set families rather than matroid rank functions. 
The following propositions (Propositions~\ref{prop:feasible}--\ref{prop:dominant}) are known facts, but
we provide alternative proofs in terms of matroid rank functions.
\begin{proposition}
	\label{prop:feasible}
	For any $X\subseteq A$, we have $\Cm(X)\in \F(\zero,r)$.
\end{proposition}
\begin{proof}
	It suffices to show $|\Cm(X)\cap B|\leq r(B)$
	for any $B\subseteq A$.
	By \eqref{eq:choice1}, we have 
	$\Cm(X)\cap B=\set{e_{i}\in A| r(A_{i}\cap X)-r(A_{i-1}\cap X)=1,~e_{i}\in B}$.
	For any $e_{i}\in B$, since  $A_{i-1}\cap X\cap B\subseteq A_{i-1}\cap X$
	and $A_{i}\cap X\cap B=(A_{i-1}\cap X\cap B)+e_{i}$,
	the diminishing returns property of $r$ implies
	\[r(A_{i}\cap X)-r(A_{i-1}\cap X)\leq r(A_{i}\cap X\cap B)-r(A_{i-1}\cap X\cap B).\]
	Thus, $e_{i}\in B$, $r(A_{i}\cap X)-r(A_{i-1}\cap X)=1$ imply
	$r(A_{i}\cap X\cap B)-r(A_{i-1}\cap X\cap B)=1$.
	Therefore,
	\begin{align*}
	|\Cm(X)\cap B|
	&=|\set{e_{i}\in A| r(A_{i}\cap X)-r(A_{i-1}\cap X)=1,~e_{i}\in B}|\\
	&\leq |\set{e_{i}\in A| r(A_{i}\cap X\cap B)-r(A_{i-1}\cap X\cap B)=1}|\\
	&=\textstyle{\sum_{i:1\leq i\leq n}} [~r(A_{i}\cap X\cap B)-r(A_{i-1}\cap X\cap B)~]=r(X\cap B).
	\end{align*} 
	The monotonicity of $r$ implies $r(X\cap B)\leq r(B)$, and the proof is completed. 
\end{proof}
\begin{proposition}
	\label{prop:rank}
	For every $X\subseteq A$ and $j=1,2,\dots, n$, we have $|\Cm(X)\cap A_{j}|=r(A_{j}\cap X)$.
	In particular, $|\Cm(X)|=r(X)$.
\end{proposition}
\begin{proof}
	By \eqref{eq:choice1}, 
	$\Cm(X)\cap A_{j}=\set{e_{i}\in A_{j}| r(A_{i}\cap X)-r(A_{i-1}\cap X)=1}$.
	This implies
	$|\Cm(X)\cap A_{i}|
	=\textstyle{\sum_{i:1\leq i\leq j}} [~r(A_{i}\cap X)-r(A_{i-1}\cap X)~]
	=r(A_{j}\cap X)-r(A_{0}\cap X)=r(A_{j}\cap X)$.
\end{proof}
\begin{proposition}
	\label{prop:substitute}
	$\Cm$ is {\bf substitutable}, i.e., 
	$X\subseteq Y\subseteq A$ implies $X\setminus \Cm(X)\subseteq Y\setminus \Cm(Y)$.
\end{proposition}
\begin{proof}
	Suppose that $e_{i}\in X\setminus \Cm(X)$ for some $i$.
	This implies $r(A_{i}\cap X)-r(A_{i-1}\cap X)=0$ by \eqref{eq:choice2}.
	By the diminishing returns property and
	$X\subseteq Y$, the value of $r(A_{i}\cap Y)-r(A_{i-1}\cap Y)$ is also $0$,
	and hence $e_{i}\in Y\setminus \Cm(Y)$ by \eqref{eq:choice2}.
\end{proof}
\begin{proposition}
	\label{prop:size-monotone}
	$\Cm$ is {\bf size-monotone}, i.e., 
	$X\subseteq Y\subseteq A$ implies $|\Cm(X)|\leq |\Cm(Y)|$.
\end{proposition}
\begin{proof}
	This immediately follows from Proposition~\ref{prop:rank}
	and the monotonicity of $r$.
\end{proof}
\begin{proposition}
	\label{prop:dominant}
	For any $X\subseteq A$, the set $\Cm(X)$ {\bf dominates} every element in $X\setminus \Cm(X)$.
	That is, the following two hold.
	\begin{itemize}
		\item For every $e\in X\setminus \Cm(X)$, we have $\Cm(X)+e\not\in \F(\zero, r)$.
		\item For every $e\in X\setminus \Cm(X)$ and $e'\in \Cm(X)$, if $e\succ e'$, then $\Cm(X)+e-e'\not\in \F(\zero, r)$.
	\end{itemize}
\end{proposition}
\begin{proof}
	Let $i$ be the index such that $e=e_{i}$, i.e., $e$ is the $i$-th best element for $\succ$.
	By Proposition~\ref{prop:rank}, we have $|\Cm(X)\cap A_{i}|=r(A_{i}\cap X)$. 
	With $\Cm(X)\subseteq X$ and $e_{i}\in X\setminus \Cm(X)$, this implies
	$|(\Cm(X)+e_{i})\cap (A_{i}\cap X)|=|\Cm(X)\cap A_{i}|+1>r(A_{i}\cap X)$, and hence
	$\Cm(X)+e_{i}\not\in \F(\zero, r)$.
	
	For the second claim, let $i'$ be the index such that $e'=e_{i'}$.
	Then, $e\succ e'$ implies $i<i'$, and hence $e_{i'}\not\in A_{i}$.
	This yields $|(\Cm(X)+e_{i}-e_{i'})\cap (A_{i}\cap X)|=|\Cm(X)\cap A_{i}|+1>r(A_{i}\cap X)$,
	which implies $\Cm(X)+e_{i}-e_{i'}\not\in \F(\zero, r)$.
\end{proof}

\subsection{Fixed-point Method for Stable Matchings on Matroids}
\label{sec:proof3}
Here we introduce the lattice fixed-point framework for stable matchings on matroids, 
studied by Fleiner \cite{Fleiner01, Fleiner03} (see also Hatfield and Milgrom \cite{HM05}).

Let $I=(D, H, E, \succ_{DH}, \{\zero, r_{h}\}_{h\in H})$ be a CSM instance such that
$r_{h}$ is a matroid rank function for each $h\in H$.
That is, each hospital has a matroidal upper quota function and no lower quota.

From $(D, E, \{\succ_{d}\}_{d\in D})$, we define doctors' joint choice function $C_{D}:2^{E}\to 2^{E}$. 
For any set $N\subseteq E$,
let $C_{D}(N)$ be the set of each doctor's best choices among $N$, i.e.,
\[C_{D}(N):=\textstyle{\bigcup_{d\in D}\set{(d,h)| h\in N(d),~h=\max_{\succ_{d}}N(d)}\quad(N\subseteq E)}.\]
From $(H, E, \{\succ_{h}\}_{h\in H}, \{r_{h}\}_{h\in H})$, we define hospitals' joint choice function
$C_{H}:2^{E}\to 2^{E}$.
First, for each hospital $h\in H$, let 
$C_{h}:2^{A(h)}\to 2^{A(h)}$ be a choice function induced from $(A(h), r_{h}, \succ_{h})$ as in Section~\ref{sec:proof2}.
Then, define  $C_{H}$ by
\[C_{H}(N):=\textstyle{\bigcup_{h\in H}\set{(d,h)| d\in N(h),~d\in C_{h}(N(h))}\quad(N\subseteq E)}.\]
Define rejection functions $R_{D}, R_{H}:2^{E}\to 2^{E}$ by
\[R_{D}(N)=N\setminus C_{D}(N),\quad R_{H}(N)=N\setminus C_{H}(N)\quad (N\subseteq E),\]
and a function $F_{I}:2^{E}\times 2^{E}\to 2^{E}\times 2^{E}$ by 
\[F_{I}(N_{D}, N_{H})=(E\setminus R_{H}(N_{H}),~ E\setminus R_{D}(N_{D}))\quad (N_{D}, N_{H}\subseteq E).\]
\begin{proposition}[Fleiner \cite{Fleiner01, Fleiner03}]
	\label{prop:fixedpoint}
	For $I=(D, H, E, \succ_{DH}, \{\zero, r_{h}\}_{h\in H})$ such that each $r_{h}$
	is a matroid rank function, if $(N_{D}, N_{H})$ is a fixed-point of $F_{I}$, then 
	$N_{D}\cap N_{H}=C_{D}(N_{D})=C_{H}(N_{H})$ holds and $N_{D}\cap N_{H}$ is a stable matching of $I$.
\end{proposition}
Let $\geq$ be a partial order defined on $2^{E}\times 2^{E}$ as
\[(N_{D}, N_{H})\geq (N'_{D}, N'_{H})\iff [N_{D}\supseteq N'_{D}, ~N_{H}\subseteq N'_{H}]\]
Recall that $C_{h}$ is substitutable for each $h\in H$, 
This implies the following property of $F_{I}$.
\begin{proposition}[Fleiner \cite{Fleiner01, Fleiner03}]
	\label{prop:monotone}
	For $I=(D, H, E, \succ_{DH}, \{\zero, r_{h}\}_{h\in H})$ such that each $r_{h}$
	is a matroid rank function, the function $F_{I}$ is monotone with respect to $\geq$.
	That is, $(N_{D}, N_{H})\geq (N'_{D}, N'_{H})$ implies $F_{I}(N_{D}, N_{H})\geq F_{I}(N'_{D}, N'_{H})$. 
\end{proposition}

The monotonicity of $F_{I}$ implies the existence of a stable matching as follows.
%
\begin{proposition}[Fleiner \cite{Fleiner01, Fleiner03}]
	\label{prop:algo}
	Let $I=(D, H, E, \succ_{DH}, \{\zero, r_{h}\}_{h\in H})$ be an instance such that each $r_{h}$
	is a matroid rank function.
	One can find a stable matching in $O(|E|\cdot {\rm EO}_{DH})$ time, where
	${\rm EO}_{DH}$ is a time to compute $C_{D}(N)$ and $C_{H}(N)$ for any $N\subseteq E$.
\end{proposition}
\begin{proof}
	Since $(E,\emptyset)$ is the maximum in $2^{E}\times 2^{E}$ with respect to $\geq$, 
	we have $(E,\emptyset)\geq F_{I}(E,\emptyset)$.
	As $F_{I}$ is monotone by Proposition~\ref{prop:monotone}, then
	\[(E,\emptyset)\geq F_{I}(E,\emptyset)\geq F_{I}(F_{I}(E,\emptyset))\geq \cdots \geq F_{I}^{k}(E,\emptyset)\geq\cdots.\]
	Since $2^{E}\times 2^{E}$ is a finite lattice whose longest chain is of length $2|E|$,  
	we have $F_{I}^{k}(E,\emptyset)=F_{I}(F_{I}^{k}(E,\emptyset))$
	for some $k\leq 2|E|$.
	Then, $(N^{\ast}_{D}, N^{\ast}_{H}):=F_{I}^{k}(E,\emptyset)$ is a fixed-point of $F_{I}$ and,
	by Proposition~\ref{prop:fixedpoint},
	$N^{\ast}_{D}\cap N^{\ast}_{H}$ is a stable matching of $I$.
\end{proof}

Fleiner also provided the following structural result on
the set of stable matchings.
\begin{proposition}[Fleiner \cite{Fleiner01, Fleiner03}]
	\label{prop:rural}
	Let $I=(D, H, E, \succ_{DH}, \{\zero, r_{h}\}_{h\in H})$ be an instance 
	such that each $r_{h}$ is a matroid rank function.
	For any two stable matchings $M, M'\subseteq E$ of $I$ and any hospital $h\in H$,
	we have $|M(h)|=|M'(h)|$. 
\end{proposition}

\subsection{Proof of Theorem \ref{thm:g-mat1}}
\label{sec:proof4}
We are now ready to show Theorem~\ref{thm:g-mat1}.
Recall that $I$ and $I'$ are defined as 
\begin{align*}
&I=(D, H, E, \succ_{DH}, \{(p_{h}, q_{h})\}_{h\in H}),\\
&I'=(D, H, E, \succ_{DH}, \{(\zero,\overline{p_{h}})\}_{h\in H}),
\end{align*}
where $(p_{h}, q_{h})$ is g-matroidal (i.e., is paramodular and satisfies $0\leq p_{h}(B)\leq q_{h}(B)\leq |B|$) for each $h\in H$.
Here, $\overline{p_{h}}$ is the complement of $p_{h}$ defined as 
$\overline{p_{h}}(B)=p_{h}(A(h))-p_{h}(A(h)\setminus B)$. 
Observe the following basic fact of a g-matroidal pair.
\begin{claim}
	\label{claim:rankfunc}
	$\overline{p_{h}}$ is a matroid rank function 
	and $\overline{p_{h}}(A(h))=p_{h}(A(h))$ for each $h\in H$.
\end{claim}
\begin{proof}
	Since $(p_{h}, q_{h})$ is g-matroidal, $p_{h}$ is supermodular
	and $0\leq p_{h}(B)\leq |B|$ for every $B\subseteq A(h)$.
	Then, for any $B'\subseteq B\subseteq A(h)$, we have
	$p_{h}(B')\leq p_{h}(B')+p_{h}(B\setminus B')\leq p_{h}(B)+p_{h}(\emptyset)=p_{h}(B)$,
	i.e., $p_{h}$ is monotone. 
	Then, $\overline{p_{h}}$ is submodular, monotone, and 
	$0\leq \overline{p_{h}}(B)\leq |B|$ for every $B\subseteq A(h)$,
	i.e., $\overline{p_{h}}$ is a matroid rank function.
	Also, $\overline{p_{h}}(A(h))=p_{h}(A(h))-p_{h}(\emptyset)=p_{h}(A(h))$.
\end{proof}
By Claim~\ref{claim:rankfunc},
Propositions~\ref{prop:algo} and \ref{prop:rural} imply the following.
\begin{lemma}
	\label{lem:2and3}
	$I'$ has a stable matching. Also, 
	for any stable matchings $M$, $M'$ of $I$ and any hospital $h\in H$,
	we have $|M(h)|=|M'(h)|$.
\end{lemma}

Lemma~\ref{lem:2and3} implies that $I'$ has a stable matching and
that conditions (b) and (c) in Theorem~\ref{thm:g-mat1} are equivalent.

What is left is to show that the condition (a) is also equivalent.
For this purpose, we prepare the following three claims.
The first and second claims are basic facts of paramodular functions \cite{Frankbook}.
The third one utilizes the exchange property of g-matroids (M$^{\natural}$-convex families).
\begin{claim}
	\label{claim:preparation1}
	For any $h\in H$ and $X\subseteq A(h)$, suppose that $|X|=\overline{p_{h}}(A(h))=p_{h}(A)$ holds.
	Then, we have $X\in \F(\zero, \overline{p_{h}})$ if and only if $X\in \F(p_{h}, q_{h})$.
\end{claim}
\begin{proof}
	We abbreviate $p_{h}$, $q_{h}$, $A(h)$ to $p$, $q$, $A$, respectively,
	and denote $\overline{B}:=A\setminus B$ for  $B\subseteq A$.
	
	To show the ``if'' part, suppose $X\in \F(p,q)$.
	Then $|X\cap \overline{B}|\geq p(\overline{B})$ for any $B\subseteq A$.
	Since $|X|=p(A)$, then $|X\cap B|=|X|-|X\cap \overline{B}|\leq p(A)-p(\overline{B})=\overline{p}(B)$.
	Thus, $X\in \F(\zero, \overline{p})$.
	
	To show the ``only if'' part, suppose $X\in \F(\zero, \overline{p})$.
	We show $p(B)\leq |X\cap B|\leq q(B)$ for any $B\subseteq A$.
	By the cross-inequality $q(B)-p(A)\geq q(B\setminus A)-p(A\setminus B)$,
	we have $q(B)\geq p(A)-p(\overline{B})$, which implies 
	$|X\cap B|\leq \overline{p}(B)=p(A)-p(\overline{B})\leq q(B)$, and thus $|X\cap B|\leq q(B)$. 
	Also, since $|X\cap \overline{B}|\leq \overline{p}(\overline{B})=p(A)-p(B)$,
	we have $|X\cap B|=|X|-|X\cap \overline{B}|\geq p(A)-(p(A)-p(B))= p(B)$.
	Thus, $|X\cap B|\geq p(B)$.
\end{proof}

Since $\overline{p_{h}}$ is a matroid rank function for each $h\in H$,
we can define the choice function 
$C_{h}:2^{A(h)}\to 2^{A(h)}$ induced from $(A(h), \overline{p_{h}},\succ_{h})$ as in Section~\ref{sec:proof3}. 
\begin{claim}
	\label{claim:preparation1.5}
	For any $h\in H$, let $C_{h}$ be the choice function induced from $(A(h), \overline{p_{h}},\succ_{h})$.
	For $Y\subseteq A(h)$ if there exists $X\in \F(p_{h}, q_{h})$ such that $X\subseteq Y$,
	then $|C_{h}(Y)|=p_{h}(A(h))$.
\end{claim}
\begin{proof}
	We abbreviate $p_{h}$, $q_{h}$, $A(h)$, $C_{h}$ to $p$, $q$, $A$, $C$, respectively.
	
	By Proposition~\ref{prop:feasible}, $C(Y)\in \F(\zero, \overline{p})$, and hence $|C(Y)|=|C(Y)\cap A|\leq \overline{p}(A)$.
	Also, Propositions~\ref{prop:rank} and \ref{prop:size-monotone} implies  $\overline{p}(X)=|C(X)|\leq |C(Y)|$.
	Since $X\in \F(p,q)$, we have $0\leq p(A\setminus X)\leq |X\cap (A\setminus X)|=0$, and hence 
	$\overline{p}(X)=p(A)-p(A\setminus X)=p(A)=\overline{p}(A)$.
	Combining these yields 
	$\overline{p}(A)\leq |C(X)|\leq \overline{p}(A)$, and hence $|C(X)|=\overline{p}(A)=p(A)$.
\end{proof}

\begin{claim}
	\label{claim:preparation2}
	For any $h\in H$, let $C_{h}$ be the choice function induced from $(A(h), \overline{p_{h}},\succ_{h})$.
	Suppose that  $X, Y\subseteq A(h)$ satisfy
	\begin{itemize}
		\setlength{\itemsep}{0mm}
		\item $X\in \F(p_{h}, q_{h})$ and $X\subseteq Y$, and 
		\item for every $d\in Y\setminus X$ and $d'\in X$,
		if $d\succ_{h} d'$, then $X+d-d'\not\in \F(p_{h}, q_{h})$.
	\end{itemize}
	Then, $C_{h}(Y)\subseteq X$.
\end{claim}
\begin{proof}
	We abbreviate $p_{h}$, $q_{h}$, $A(h)$, $C_{h}$ to $p$, $q$, $A$, $C$, respectively.
	
	By Claim~\ref{claim:preparation1.5}, $X\in \F(p,q)$ and $X\subseteq Y$ imply 
	$|C(Y)|=p(A)=\overline{p}(A)$. Also, $C(Y)\in \F(\zero, \overline{p})$ by Proposition~\ref{prop:feasible}.
	Then, Claim~\ref{claim:preparation1} implies $C(Y)\in \F(p,q)$.
	Thus, $X, C(Y)\in \F(p,q)$.
	Suppose, to the contrary, $C(Y)\subsetneq X$. Then there is some $d\in C(Y)\setminus X$.
	By the symmetric exchange axiom {\bf (B$^{\natural}$-EXC)} for $C(X)$, $Y$, and $d$, we have either
	(i) $C(Y)-d, ~X+d\in \F(p,q)$, or 
	(ii) $\exists d'\in X\setminus C(Y): C(Y)-d+d',~X+d-d'\in \F(p,q)$.
	Note that (i) cannot hold since $C(Y)-d\not\in \F(p,q)$ follows from $|C(Y)-d|<|C(Y)|=p(A)$.
	Then, (ii) holds, i.e., 
	there exists $d'\in X\setminus C(Y)$ such that $C(Y)-d+d',~X+d-d'\in \F(p,q)$.
	
	By $d\in C(Y)\setminus X\subseteq Y\setminus X$ and $d'\in X$ and $X+d-d'\in \F(p,q)$
	the assumption on $Y$ implies $d'\succ_{h} d$.
	On the other hand,
	by $|C(Y)-d+d'|=|C(Y)|=p(A)$, Proposition~\ref{claim:preparation1} 
	and $C(Y)-d+d'\in \F(p,q)$ imply
	$C(Y)-d+d'\in \F(\zero, \overline{p})$. 
	As $d\in C(Y)\setminus X \subseteq C(Y)$ and $d'\in X\setminus C(Y)\subseteq Y\setminus C(Y)$, 
	this implies $d\succ_{h} d'$ by Proposition~\ref{prop:dominant}, a contradiction.
\end{proof}

We now complete the proof of Theorem~\ref{thm:g-mat1}
by showing the following lemma,
which states the equivalence between conditions (a) and (b) in Theorem~\ref{thm:g-mat1}.
\begin{lemma}
	\label{lem:1and2}
	$I$ has an envy-free matching if and only if 
	some stable matching $M'$ of $I'$ satisfies 
	$|M'(h)|=p_{h}(A(h))$ for all $h\in H$.
\end{lemma}
\begin{proof}
	{\bf The ``if'' part:}
	Let $M'$ be a stable matching of $I'$ such that
	$|M'(h)|=p_{h}(A(h))$ for all $h\in H$.
	We show that $M'$ is also an envy-free matching of $I$.
	
	As $M'$ is  feasible for $I'$, we have $|M'(d)|\leq 1$ for every $d\in D$ and
	$M'(h)\in \F(\zero, \overline{p_{h}})$ for every $h\in H$.
	By Claim~\ref{claim:preparation1} and $|M'(h)|=p_{h}(A(h))$, 
	then $M'(h)\in \F(p_{h}, q_{h})$,
	and hence $M'$ is also a matching in $I$.
	Suppose, to the contrary, that there is a doctor $d\in D$ who has
	justified envy toward $d'\in D$ with $M'(d')=h$.
	Then, (i) $d$ is unassigned or $h\succ_{h} M'(d)$ and 
	(ii) $M'(h)+d-d'\in \F(p_{h}, q_{h})$ and $d\succ_{h} d'$.
	Note that $|M'(h)+d-d'|=|M'(h)|=p_{h}(A(h))$.
	Then, Claim~\ref{claim:preparation1} implies $M'(h)+d-d'\in \F(\zero, \overline{p_{h}})$.
	This means that $(d,h)$ is a blocking pair for $M'$ in $I'$, a contradiction.
	
	\medskip
	{\bf The ``only if'' part:}
	Suppose that $I$ has an envy-free matching $M$.
	We now construct a stable matching $M'$ of $I'$ 
	satisfying $|M'(h)|=p_{h}(A(h))$ for all $h\in H$.
	
	For $I'=(D, H, E, \succ_{DH}, \{(\zero,\overline{p_{h}})\}_{h\in H})$,
	define 
	$C_{D}, C_{H}:2^{E}\to 2^{E}$ as in Section~\ref{sec:proof3}. 
	That is, $C_{D}$ returns the set of each doctor's best choices 
	and $C_{H}$ is defined by combining $\{C_{h}\}_{h\in H}$, 
	where $C_{h}$ is induced from $(A(h), \overline{p_{h}},\succ_{h})$.
	From $C_{D}$ and $C_{H}$, we define $F_{I'}:2^{E}\times 2^{E}\to 2^{E}\times 2^{E}$ 
	as in Section~\ref{sec:proof3}. 
	Define two supersets $N_{D}, N_{H}\subseteq E$ of $M$ by
	\begin{align*}
	&N_{D}:=M\cup\set{(d,h)\in E\setminus M| M(d)\succ_{d} h},\\
	&N_{H}:=M\cup (E\setminus N_{D}).
	\end{align*}

	Note that $N_{H}\setminus M=E\setminus N_{D}$, and hence
	every $(d, h)\subseteq N_{H}\setminus M$ satisfies $(d,h)\not\in N_{D}$, and hence $h\succ_{d} M(d)$.
	Since $M$ is an envy-free matching, then for every $d'\in M(h)$ with $d\succ_{h} d'$ we have $M(h)+d-d'\not\in \F(p_{h}, q_{h})$,
	since otherwise $d$ has justified envy toward $d'$.
	Thus, we have
	\begin{itemize}
		\setlength{\itemsep}{0mm}
		\item $M(h)\in \F(p_{h}, q_{h})$ and $M(h)\subseteq N_{H}(h)$, and 
		\item for every $d\in N_{H}(h)\setminus M(h)$ and $d'\in M(h)$, if $d\succ_{h} d'$, then $M(h)+d-d'\not\in \F(p_{h}, q_{h})$.
	\end{itemize}
	Claim~\ref{claim:preparation2} then implies $C_{h}(N_{H}(h))\subseteq M(h)$ for each $h\in H$, and hence $C_{H}(N_{H})\subseteq M$.
	\begin{equation}
	\label{eq:ND}
	E\setminus R_{H}(N_{H})=(E\setminus N_{H})\cup C_{H}(N_{H})\subseteq (E\setminus N_{H})\cup M=N_{D}.
	\end{equation}
	Also, by the definition of $C_{D}$ and $N_{D}$, we have $C_{D}(N_{D})=M$, which implies
	\begin{equation}
	\label{eq:NH}
	E\setminus R_{D}(N_{D})=(E\setminus N_{D})\cup C_{D}(N_{D})=(E\setminus N_{D})\cup M=N_{H}.
	\end{equation}
	Recall the partial order $\geq$ defined on $2^{E}\times 2^{E}$ in Section~\ref{sec:proof3}.
	By \eqref{eq:ND} and \eqref{eq:NH}, we have
	\[(N_{D}, N_{H})\geq(E\setminus R_{H}(N_{H}), E\setminus R_{D}(N_{D}))=F_{I'}(N_{D}, N_{H}).\]
	Since $F_{I'}$ is monotone by Proposition~\ref{prop:monotone}, this implies
	\[(N_{D}, N_{H})\geq F_{I'}(N_{D}, N_{H})\geq F_{I'}(F_{I'}(N_{D}, N_{H}))\geq \cdots \geq F_{I'}^{k}(N_{D}, N_{H})\geq\cdots,\]
	and hence there is $k$ such that $F_{I'}^{k}(N_{D}, N_{H})$ is a fixed-point of $F_{I'}$.
	Denote it by $(N^{k}_{D}, N^{k}_{D})$ and define $M':=C_{H}(N^{k}_{H})$.
	By Proposition~\ref{prop:fixedpoint},  $M'$ is a stable matching of $I'$.
	
	What is left is to show $|M'(h)|=p_{h}(A(h))$ for all $h\in H$.
	Since $(N_{D}, N_{H})\geq F_{I'}^{k}(N_{D}, N_{H})=(N^{k}_{D}, N^{k}_{D})$, we have
	$N_{H}\subseteq N^{k}_{H}$, 
	Then $M\subseteq C_{H}(N_{H})\subseteq N_{H}\subseteq N^{k}_{H}$,
	and hence $M(h)\subseteq N^{k}_{H}(h)$ for each $h\in H$.
	By $M(h)\in \F(p_{h}, q_{h})$ and Claim~\ref{claim:preparation1.5}, 	$|M'(h)|=|C_{h}(N^{k}_{H}(h))|=p_{h}(A(h))$.
\end{proof}

Combining Lemmas~\ref{lem:2and3} and \ref{lem:1and2}
completes the proof of Theorem~\ref{thm:g-mat1}.

\subsection{Proof of Theorem \ref{thm:g-mat2}}
\label{sec:proof5}
We first show that  the ``while loop'' of  the algorithm {\sf EF-Paramodular-CSM} computes  
a stable matching of $I'=(D, H, E, \succ_{DH}, \{(\zero,\overline{p_{h}})\}_{h\in H})$.
By the proof of Proposition~\ref{prop:algo},
it suffices to show that, each iteration updates $(N_{D}, N_{H})$ to $F_{I'}(N_{D}, N_{H})$.
That is,  we show that the subsets $R_{D}$ and $R_{H}$ defined as
\begin{align*}
&\textstyle{R_{D}:= \bigcup_{d\in D} \set{(d,h) |h\in N_{D}(d),~h\neq \max_{\succ_{d}}N_{D}(d)},}\\
&\textstyle{R_{H}:= \bigcup_{h\in H} \set{(d,h)| d\in N_{H}(h),~ p_{h}(A(h)\setminus N_{H}(h)_{\succeq_{h} d})=p_{h}(A(h)\setminus N_{H}(h)_{\succ_{h} d})}}
\end{align*}
coincide with $N_{D}\setminus C_{D}(N_{D})$ and $N_{H}\setminus C_{H}(N_{H})$, respectively,
where $C_{D}$ and $C_{H}$ are defined for $I'$ as in Section~\ref{sec:proof3}.
By definition, $R_{D}=N_{D}\setminus C_{D}(N_{D})$ can be checked easily.
To show $R_{H}=N_{H}\setminus C_{H}(N_{H})$, recall the definition of $C_{H}$ in Section~\ref{sec:proof3}.
\[C_{H}(N)=\textstyle{\bigcup_{h\in H}\set{(d,h)| d\in N(h),~d\in C_{h}(N(h))}}\quad (N\subseteq E).\]
Here, each $C_{h}:2^{A(h)}\to 2^{A(h)}$ is a choice function induced from $(A(h), \overline{p_{h}}, \succ_{h})$. 
By definitions of $C_{h}$ and $\overline{p_{h}}$, for any $N\subseteq E$, we have
\begin{align*}
C_{h}(N(h))
&=\set{d\in N(h)| \overline{p_{h}}(N(h)_{\succeq_{h} d})>\overline{p_{h}}(N(h)_{\succ_{h} d})},\\
&=\set{d\in N(h)| p_{h}(A(h)\setminus N(h)_{\succeq_{h} d})<p_{h}(A(h)\setminus N(h)_{\succ_{h} d})}.
\end{align*}
By the monotonicity of $p_{h}$ (shown in the proof of Claim~\ref{claim:rankfunc}),
for any $d\in N(h)$, we have 
$p_{h}(A(h)\setminus N(h)_{\succeq_{h} d})\leq p_{h}(A(h)\setminus N(h)_{\succ_{h} d})$.
Then, for any $h\in H$, $N\subseteq E$, and $d\in N(h)$, 
\[d\in N(h)\setminus C_{h}(N(h)) \iff p_{h}(A(h)\setminus N(h)_{\succeq_{h} d})=p_{h}(A(h)\setminus N(h)_{\succ_{h} d}).\]
Thus, we have $R_{H}=N_{H}\setminus C_{H}(N_{H})$.

We now analyze the time complexity.
As shown in the proof of Proposition~\ref{prop:algo},
a stable matching is found by computing $F_{I'}$ at most $2|E|$ times,
i.e., the ``while loop'' is iterated $O(|E|)$ times.
Also, we see that each iteration can be done in $O(|E|)$ time.
Checking the condition $|M'(h)|=p_{h}(A(h))~(h\in H)$ is done in $O(|E|)$ time.
Thus, the algorithm runs in $O(|E|^2)$ time.
\end{document}